\documentclass{amsart}
\usepackage{amsmath, amsfonts, amssymb, amsthm, xcolor, amscd, latexsym,xspace,enumerate,newlfont, amscd, graphicx,  adjustbox, tikz-cd,pdfpages,tikz}
\usepackage[mathscr]{eucal}
\usepackage[pagebackref, colorlinks=true,linkcolor=blue,citecolor=red]{hyperref}
\usepackage{hyperref}

\newtheorem{theorem}{Theorem}[section]
\newtheorem{corollary}[theorem]{Corollary}
\newtheorem{lemma}[theorem]{Lemma}
\newtheorem{conj}[theorem]{Conjecture}

\theoremstyle{definition}
\newtheorem{definition}[theorem]{Definition}
\newtheorem{remark}[theorem]{Remark}

\numberwithin{equation}{section}

\newcommand{\be}{\begin{equation}}
\newcommand{\en}{\end{equation}}

\newcommand{\Hil}{{\cal H}}
\newcommand{\1}{{1\!\!1}}
\newcommand{\F}{{\cal F}}

\usepackage{vmargin}

\def\lead{\leaders\hbox to 1.5ex{\hss${.}$\hss}\hfill}
\def\arr{\hbox to 40pt{\rightarrowfill}}
\def\larr{\hbox to 40pt{\leftarrowfill}}

\long\def\alert#1{\parindent2em\smallskip\hbox to\hsize%
{\hskip\parindent\vrule%
\vbox{\advance\hsize-2\parindent\hrule\smallskip\parindent.4\parindent%
\narrower\noindent#1\smallskip\hrule}\vrule\hfill}\smallskip\parindent0pt}

\begin{document}

\title[Topological decompositions of the Pauli group and...]{Topological decompositions of the Pauli group and their influence on  dynamical systems}
\author[F. Bagarello]{Fabio Bagarello}
\address{Fabio Bagarello\endgraf
Dipartimento di  Ingegneria\endgraf 
Universit\'a  di Palermo\endgraf
Viale delle Scienze, I-90128, Palermo, Italy\endgraf
and \endgraf
Istituto Nazionale di Fisica Nucleare\endgraf
Sezione di Napoli, Via Cinthia, I-80126, Napoli, Italy}
\email{fabio.bagarello@unipa.it}

\author[Y. Bavuma]{Yanga Bavuma}
\address{Yanga Bavuma\endgraf
Department of Mathematics and Applied Mathematics\endgraf
University of Cape Town\endgraf
Private Bag X1, 7701, Rondebosch, Cape Town, South Africa}
\email{BVMYAN001@myuct.ac.za}

\author[F.G. Russo]{Francesco G. Russo}
\address{Francesco G. Russo\endgraf
Department of Mathematics and Applied Mathematics\endgraf
University of Cape Town\endgraf
Private Bag X1, 7701, Rondebosch, Cape Town, South Africa\endgraf
}
\email{francescog.russo@yahoo.com}


\date{\today}


\begin{abstract}
In the present paper we show that it is possible to obtain the well known Pauli group $P=\langle X,Y,Z \ | \ X^2=Y^2=Z^2=1, (YZ)^4=(ZX)^4=(XY)^4=1 \rangle $ of order $16$ as an appropriate quotient group of two distinct spaces of orbits of the three dimensional sphere $S^3$. The first of these spaces of orbits is realized via an action of the quaternion group $Q_8$ on $S^3$; the second one via an action of the cyclic group of order four $\mathbb{Z}(4)$  on $S^3$. We deduce a result of decomposition of $P$ of topological nature and then we find, in connection with the theory of pseudo-fermions, a possible physical interpretation of this decomposition. 
\end{abstract}

\date{\today}

\maketitle

\section{Statement of the main result}

The Pauli group $P$ is a finite group of order $16$, introduced by W. Pauli in \cite{pauli}, and it is an interesting $2$-group, which has relevant properties for dynamical systems and theoretical physics. Pauli matrices are 
\begin{equation}\label{paulimatricesxyz} X=\left(
\begin{array}{rr}
0 & 1 \\
1 & 0%
\end{array}%
\right) ,Y=\left(
\begin{array}{rr}
0 & -i \\
i & 0%
\end{array}%
\right) \text{ and }Z=\left(
\begin{array}{rr}
1 & 0 \\
0 & -1
\end{array}%
\right),
\end{equation}
and one can check that $
X^2=Y^{2}=Z^{2}=\1=\left(
\begin{array}{cc}
1 & 0 \\
0 & 1%
\end{array}%
\right) $ and in addition that $$(YZ)^4=(ZX)^4=(XY)^4=\1,  \ \ {(XYZ)}^4=[XYZ,X]=[XYZ,Y]=[XYZ,Z]=\1,$$
in fact
\begin{equation}\label{presentation1}
P=\langle X,Y,Z \ | \ X^2=Y^2=Z^2=1, (YZ)^4=(ZX)^4=(XY)^4=1 \rangle. 
\end{equation}
In Quantum Mechanics the role of $P$ is well known (see \cite{kibler, messiah}) and it allows us to detect symmetries in numerous dynamical systems. More recently,  Pauli groups have been studied in connection with their rich lattice of subgroups (especially  abelian subgroups) and several applications have been found in Quantum Information Theory. Notable examples are quantum error correcting codes and the problem of finding mutually unbiased bases (see \cite{ gottesman1997stabilizer, gottesman1999, knill, KW, rocchetto}). It might be useful to observe that Pauli groups and Heisenberg groups have a precise meaning for their  generators and relations in physics. In particular, the role of Heisenberg algebras has been recently explored in   \cite{bagrus1, bagrus2, bagrus3}  within the framework of pseudo-bosonic operators.

Here we will describe $P$ in terms of an appropriate quotient of the fundamental group of a topological space, identifying in this way $P$ from a geometric point of view. A direct  geometric construction of $P$ is one of our main contributions. We involve some methods of general nature, but  develop  a series of tools which are designed for $P$ only. This choice is made for a specific motivation: we want to  avoid a universal approach for the notions of amalgamated product and \textit{central product} (see \cite{DoerHawk, gorenstein, hatcher,  kosnio}, or Section 3 later on), even if these two notions may be formalized in category theory, or  in classes of finite groups which are larger than the class of  $2$-groups. Our approach will have the advantage to analyse directly $P$, involving low dimensional topology and combinatorial results for which we do not need a computational software.

Following \cite{hatcher, kosnio}, $\pi(X)$ denotes the fundamental group of a path connected topological space $X$ and  $X/G$ the space of orbits of $X$ under a (left) action of a group $G$ acting on $X$. For the $n$-dimensional sphere $S^n$ of the euclidean space $\mathbb{R}^{n+1}$ we recall that  $S^n=\partial B^{n+1}_1(0)$, that is, $S^n$ agrees with the boundary of the ball $B^{n+1}_1(0)$ of center at the origin and radius one in $\mathbb{R}^{n+1}$. The terminology is standard and follows \cite{DoerHawk, gorenstein, hatcher, kosnio, MaKaSo}. In particular, a manifold $M$  of dimension $\mathrm{dim}(M)=n$  is a Hausdorff space $M$ in which each point of $M$ has an open neighbourhood homeomorphic to to the open ball  $B^n_1(0)$, that is, to $ B^n_1(0) \setminus \partial B^n_1(0)$ (see \cite[Definition 11.1]{kosnio}). Note that $Q_8$ denotes the well known quaternion group of order $8$ and $\mathbb{Z}(m)$ the cyclic group of integers modulo $m$ (with $m$ positive integer). On the other hand, we will refer to the usual \textit{connected sum} $\#$ between manifolds (see \cite{hatcher, kosnio}). Our first main result is the following.

\begin{theorem}\label{maintheorem}
There exist two compact path connected orbit spaces $U=S^3/Q_8$ and  $V=S^3/\mathbb{Z}(4)$ such that the following conditions hold:
\begin{itemize}
\item[(i)]$U \cup V$ is a compact path connected space with  $U \cap V \neq \emptyset$, $\pi (U \cap V) $ cyclic of order $2$ and $P \simeq \pi(U \cup V)/N$ for some normal subgroup $N$ of $\pi(U \cup V)$; 
\item[(ii)]$U \# V$ is a Riemannian manifold of $\mathrm{dim} (U \# V)=3$  and $P \simeq \pi(U \# V)/L$ for some normal subgroup $L$ of $\pi(U \# V)$.
\end{itemize}
Both in case (i) and (ii), $P$ is  central product of $\pi (U)$ and  $\pi (V)$.
\end{theorem}

A separation in (i) and (ii)  is made, because we  stress that a topological decomposition  $U \cup V$  cannot produce enough information on the dimension of $U \# V$, or on the fact that it is Riemannian.  The formation of $U \cup V$ is in fact more general than the formation of $U \# V$. On the other hand, the algebraic decomposition of $P$ is not affected from the topological ones in  Theorem \ref{maintheorem}.
  
Another interesting result of the present paper is that $P$ can be expressed in terms of \textit{pseudo-fermionic operators}. The reader can refer to \cite{baginbagbook, bagpf1, bagpf2} (or  to a final Appendix here) for  the main notions on the theory of pseudo-fermionic operators: these are  operators defined by suitable anti-commutation relations. It is relevant to note that some dynamical aspects of physical systems (involving pseudo-fermionic operators) are connected with a suitable decomposition of $P$. This will be shown later, and is stated by the following theorem.

\begin{theorem}\label{maintheorembis}
There are two  dynamical systems $S$ and $T$ involving pseudo-fermions with  groups of symmetries respectively $P_\mu \simeq P$ and $Q_8$ but with the same hamiltonian $H_S=H_T$.  In particular, there exist dynamical systems admitting larger groups of symmetries, whose size does not affect  the dynamical aspects of the system.
\end{theorem}

The paper is structured as follows. We begin to recall some facts of Algebraic Topology in Section 2, where  finite group actions are involved on three dimensional spheres. Then Section 3 is devoted to prove some results on  central products in combinatorial group theory, separating a paragraph of general nature from one which is specific on the behaviour of the Pauli group. Finally, the proofs of our main results are placed in Sections 4 and 5. In particular, connections with mathematical physics are presented in Section 5 and  conclusions are placed in Section 6. A brief review on pseudo-fermions is given in  Appendix.

\section{Some facts on the actions of $Q_8$ and $\mathbb{Z}(4)$ on $S^3$}

Let $K$ be a field and $A$ be a vector space over $K$ with an additional internal operation  $$\bullet : (\textbf{x}, \textbf{y}) \in A \times A \mapsto  \textbf{x} \bullet \textbf{y} \in A$$ such that 
$$(\textbf{x}+\textbf{y}) \bullet \textbf{z}=\textbf{x} \bullet \textbf{z}+\textbf{y} \bullet \textbf{z}, \ \ \textbf{x} \bullet (\textbf{y} + \textbf{z})=\textbf{x} \bullet \textbf{y}+\textbf{x} \bullet \textbf{z}, \ \ (a \textbf{x}) \bullet (b \textbf{y})=(ab)(\textbf{x} \bullet \textbf{y})$$ for all $\textbf{x}, \textbf{y}, \textbf{z} \in A$ and $a,b \in K$. In this situation $A$ is an \textit{algebra} on $K$. The quaternion algebra $\mathbb{H}$ (on $\mathbb{R}$) is a different way to endow an algebraic (and topological) structure on the usual euclidean space $\mathbb{R}^4$. Specifically it consists of all the linear combinations $a\textbf{1}+b\textbf{i}+c\textbf{j}+d\textbf{k}$, where $a,b,c, d \in \mathbb{R}$ and $\{\textbf{1}, \textbf{i}, \textbf{j}, \textbf{k}\}$ forms a standard basis for the vector space $\mathbb{H}$, hence $$ \mathrm{span}(\textbf{1}, \textbf{i}, \textbf{j}, \textbf{k})=\mathrm{span}((1,0,0,0), (0,1,0,0), (0,0,1,0), (0,0,0,1))=\mathbb{H}.$$  In addition to the pointwise sum of elements of $\mathbb{H}$ and to the usual scalar multiplication on $\mathbb{R}$, we have the (internal) multiplication in $\mathbb{H}$, given by
$$\textbf{x} \bullet \textbf{y}=(a_1\textbf{1}+b_1\textbf{i}+c_1\textbf{j}+d_1\textbf{k}) \bullet (a_2\textbf{1}+b_2\textbf{i}+c_2\textbf{j}+d_2\textbf{k})=(a_1 a_2-b_1 b_2 -c_1 c_2 -d_1 d_2)\textbf{1} $$
$$+(a_1 b_2+b_1 a_2 +c_1 d_2 -d_1 c_2)\textbf{i} +(a_1 c_2-b_1 d_2 +c_1 a_2 +d_1 b_2)\textbf{j}+(a_1 d_2+b_1 c_2 -c_1 b_2 +d_1 a_2)\textbf{k}.$$
From this rule, one can check that 
\begin{equation}\label{rulesq8}
\textbf{i} \bullet \textbf{i}=\textbf{i}^2=-\textbf{1}, \ \ \textbf{j} \bullet \textbf{j}=\textbf{j}^2=-\textbf{1}, \ \ \textbf{k} \bullet \textbf{k}=\textbf{k}^2=-\textbf{1}, \ \ \textbf{i} \bullet \textbf{j}=\textbf{k}, \ \ \textbf{j} \bullet \textbf{k}=\textbf{i}, \ \ \textbf{k} \bullet \textbf{i}=\textbf{j}.
\end{equation}
and that every nonzero element of $\mathbb{H}$ has an inverse (w.r.t. $\bullet$ ) of the form
$$(a\textbf{1}+b\textbf{i}+c\textbf{j}+d\textbf{k})^{-1}=\frac{1}{a^2+b^2+c^2+d^2} (a\textbf{1}-b\textbf{i}-c\textbf{j}-d\textbf{k}),$$
and so every nonzero element of $\mathbb{H}$ has a multiplicative inverse. In this context the usual sphere $S^3=\{(x,y,z,t) \ | \ x^2+z^2+y^2+t^2=1\} \subseteq \mathbb{R}^4$  can be regarded as the set of elements of $\mathbb{H}$ with norm equal to one. On $\mathbb{H}$ we may introduce the norm $\| \ \|$  which is defined as the usual Euclidean norm in $\mathbb{R}^4$, but on $\mathbb{H}$ now the norm $\| \ \|$ becomes multiplicative, i.e. $$\|\textbf{x} \bullet \textbf{y}\|=\|\textbf{x}\| \ \  \|\textbf{y}\|  \ \ \mbox{for all} \ \ \textbf{x},\textbf{y} \in \mathbb{H},$$ so the multiplication of vectors whose norms equal one will result in a vector whose norm equals one. From \cite{hatcher}, we know that such a norm allows us to give a topological  structure on $\mathbb{H}$  and  $S^3$ may be regarded as topological subspace of $\mathbb{H}$, but also as a group with respect to the algebraic structure of $\mathbb{H}$, because the operation $$(\textbf{x},\textbf{y}) \in S^3 \times S^3 \longmapsto  \textbf{x} \bullet \textbf{y} \in S^3$$ is well defined and endow $S^3$ of a structure of group.

Let recall some classical notions from \cite{hatcher, kosnio}. Given a group  $G$ and  a set $X$,  the map $(g,x) \in G \times X \longmapsto g \cdot x \in X$ is said to be a (left) \textit{action} if $1 \cdot x=x$ for all $x \in X$ and if $g \cdot (h \cdot x)=(g \cdot h)\cdot x$ for all $x \in X$ and $g,h \in G$. Then we say that  $G$ acts \textit{freely} on $X$ (or is a \textit{free action}) if $g \cdot x \neq x$ for all $x \in X$, $g \in G$, $g \neq 1$. The following  fact is known, but we offer a direct argument:

\begin{lemma}\label{superlegalaction} The groups $Q_8$ and $\mathbb{Z}(4)$ act freely on $S^3$.
\end{lemma}

\begin{proof}

The map $(\textbf{x},\textbf{y}) \in S^3 \times S^3 \longmapsto  \textbf{x} \bullet \textbf{y} \in S^3$ makes $S^3$ a group, and  an appropriate restriction of this map to a subgroup $G$ of $S^3$ allows us to define an action on $S^3$ of the form 
\begin{equation}\label{action1}
(g,x) \in G \times S^3 \longmapsto g \cdot x=g \bullet x \in S^3.\end{equation} 
In particular we look at \eqref{rulesq8} and note that this happens when $G$ is chosen as $$Q_8=\langle \textbf{1},  \textbf{i},  \textbf{j}, \textbf{k} \ | \ \textbf{i}^2=\textbf{j}^2=\textbf{k}^2=-\textbf{1}, \ \ \textbf{i} \bullet \textbf{j}=\textbf{k}, \ \ \textbf{j} \bullet \textbf{k}=\textbf{i}, \ \ \textbf{k} \bullet \textbf{i}=\textbf{j} \ \rangle,$$ producing the action 
\begin{equation}\label{action2}
(q,x) \in Q_8 \times S^3 \longmapsto q \cdot x=q \bullet x \in S^3. \end{equation} 
  The rest follows from the fact that $\mathbb{H} \setminus \{0\}$ contains inverses for all its elements and so the equation $g \bullet x=x$ is only true for $g=1$ or $x=0$ where $g,x \in \mathbb{H}$. This shows that $Q_8$ acts freely on $S^3$. In order to show the second part of the result, we consider $$h: (z_0,z_1) \in S^3 \longmapsto (e^{\frac{2 \pi i}{4}} z_0,e^{\frac{2 \pi i 3}{4}} z_1) \in S^3$$ and observe that $S^3$, regarded as a group,  contains cyclic subgroups of order four, so we may define the action of $\mathbb{Z}(4)$ on $S^3$ by 
\begin{equation}\label{action3}(n,(z_0,z_1)) \in \mathbb{Z}(4) \times S^3 \longmapsto n \cdot (z_0,z_1)=h^n(z_0,z_1) \in S^3.\end{equation}
We need to show that the only element that fixes points in this action is the identity element. This can be checked easily, since
$$h^n(z_0,z_1)=(z_0,z_1) \ \ \Longleftrightarrow \ (e^{\frac{2 \pi i n}{4}} z_0,e^{\frac{2 \pi i 3 n}{4}} z_1)=(z_0,z_1) \ \ \Longleftrightarrow \ e^{\frac{2 \pi i n}{4}} z_0=z_0,  \ \ e^{\frac{2 \pi i 3 n}{4}} z_1=z_1.$$
Since $z_0$ or $z_1$ are different from zero, we have for all $k \ge 0$
$$ e^{\frac{2 \pi i n}{4}} z_0=z_0 \ \Longleftrightarrow \  \frac{2 \pi i n}{4}=0+2 \pi i k \ \Longleftrightarrow   \ n= 4k;$$
$$ e^{\frac{2 \pi i 3 n}{4}} z_1=z_1 \ \Longleftrightarrow \  \frac{2 \pi i 3 n}{4}=0+2 \pi i k   \ \Longleftrightarrow   \   2 \pi i 3 n=8 \pi i k. $$
Clearly $n \equiv 0 \mod 4$ and therefore the only fixed points are under the identity  action of $\mathbb{Z}(4)$. This means that  $\mathbb{Z}(4)$ acts freely on $S^3$.
\end{proof}

The  notion of \textit{properly discontinuous action} is well known and can be found in \cite[P.143]{kosnio}, namely if $X$ is a $G$-space, that is, $X$ is a topological space possessing a left action $\cdot$ of a group $G$ on $X$ and in addition the function $\theta_g : x\in X \to g\cdot x \in X$ is continuous for all $g \in G$, we say that the action $\cdot$ is properly discontinuous if for each $x \in X$ there is an open neighbourhood $V$ of $x$ such that $g_1 \cdot V \cap g_2 \cdot V = \emptyset$ for all $g_1, g_2 \in G$ with $g_1 \neq g_2$.  

\begin{corollary}\label{orbitspaces}
There are properly discontinuous actions of $Q_8$ and $\mathbb{Z}(4)$ on $S^3$. Moreover, the orbit spaces $S^3/Q_8$ and $S^3/\mathbb{Z}(4)$ are compact and path connected.
\end{corollary}

\begin{proof}
From \cite[Theorem 17.2]{kosnio} if we have a finite group $G$ that acts freely on a Hausdorff space $X$ then the action of $G$ on $X$ is properly discontinuous. From \cite[Theorem 7.8]{kosnio} and \cite[Theorem 12.4]{kosnio} the image of a compact, path connected space is itself compact and path connected.
\end{proof}

Fundamental groups of spaces of orbits can be easily computed when the actions are given. Details can be found in \cite[Chapters 18 and 19]{kosnio}. Therefore we may conclude that

\begin{lemma}\label{superlegalactiontwo}
The fundamental group $\pi (S^3/Q_8)$ of the space of orbits $S^3/Q_8$ via the action \eqref{action2} is isomorphic to $Q_8$.
Moreover $\pi (S^3/\mathbb{Z}(4)) \cong \mathbb{Z}(4)$ via the action \eqref{action3}.
\end{lemma}

\begin{proof}
We may apply \cite[Theorem 19.4]{kosnio}, that is, we have an orbit space $X/G$ produced by a properly discontinuous action of a group $G$ on a simply connected space $X$, then the fundamental group of the orbit space $X/G$ is isomorphic to the underlying group $G$.
\end{proof}

Groups acting on spheres deserve attention in literature. There are two conditions that a finite group $G$ acting freely on $S^n$ must satisfy, according to \cite[Page 75]{hatcher}. 

\begin{remark} \label{rem1}
The following  conditions are well known 
\begin{itemize}
\item[(a)] Every abelian subgroup of $G$ is cyclic. This is equivalent to saying that $G$ contains no subgroup $\mathbb{Z}(p) \times \mathbb{Z}(p)$ with $p$ prime.
\item[(b)] $G$ contains at most one element of order $2$.
\end{itemize}
\end{remark}

Because of Remark \ref{rem1}, it is not surprising that $\mathbb{Z}(4)$ acts freely on $S^3$ since it is a finite cyclic group of order four. On the other hand,  it is useful to note that:

\begin{remark} \label{rem2}
The dihedral group $D_8$ of order $8$ does not act freely on $S^3$ since it fails on condition (b) of Remark \ref{rem1}: There are in fact involutions in $D_8$. Structurally $D_8$ is very similar to $Q_8$, e.g. $[Q_8,Q_8] \cong [D_8,D_8] \cong \mathbb{Z}(2)$, $Z(Q_8) \cong Z(D_8) \cong \mathbb{Z}(2)$ and $Q_8/Z(Q_8) \cong D_8/Z(D_8) \cong \mathbb{Z}(2) \times \mathbb{Z}(2)$ but the presence of more involutions makes the difference between $Q_8$ and $D_8$.
\end{remark}

In fact Zimmermann and others \cite{zim3, zim1, zim2, zim4, zim5} developed the theory of  group actions on spheres in a series of fundamental contributions,  illustrating the properties that finite (or even infinite) groups must have in order to act on spheres of low dimension, i.e. $S^2$, $S^3$ and $S^4$.

\begin{remark} \label{rem3}It is useful to know that \cite{zim3, zim2, zim4, zim5}  explores the orientation-preserving topological actions on $S^3$ (and on the euclidean space $\mathbb{R}^3$), while  \cite{zim5} deals specifically with the following question: Is there a finite group $G$ which admits a faithful topological or smooth action on a sphere $S^d$ (of dimension $d$) but does not admit a faithful, linear action on $S^d$ ?  For each dimension $d>5$, there is indeed a finite group $G$ which admits a faithful topological action on $S^d$ (but $G$ is not isomorphic to a subgroup of the real  orthogonal group $O(d+1)$, see again \cite{zim5}). 
\end{remark}

Remarks \ref{rem1}, \ref{rem2} and \ref{rem3} show that the theory of actions of groups on spheres may be developed more generally than what is presented here for $S^3$. We indeed use this theory for our scopes.

\section{Central products and subdirect products}

We introduce some results of general nature on central products of groups, dividing the present section in two parts, one with more emphasis on the general constructions, and another one which is specific for the Pauli group.

\subsection{Classical properties of central products of groups} We need to fix some notation, mentioned in \cite{kosnio}. The inclusion map $i: U \to X$, where $U \subseteq X$, is the map that takes a point $x \in U$ to itself in $X$. And this can induce a homomorphism on fundamental groups that we denote by $i_*$. In the language of combinatorial group theory, if $X$ is a topological space; $U$ and $V$ are open, path connected subspaces of $X$; $U\cap V$ is nonempty and path-connected; $ w\in U\cap V$; then the natural inclusions
$i_1 : U \cap V \to U$, $i_2 : U \cap V \to V$, $j_1 : U \to X$ and $j_2 : V \to X$ for the following commutative diagram
 \[\begin{CD}
U \cap V @>i_1>> U \\
@Vi_2VV @VVj_1V\\
V @>j_2>> X  \\
\end{CD}\]
that induces another commutative diagram on the corresponding fundamental groups, given by 
 \[\begin{CD}
\pi(U \cap V,w) @>{(i_1)}_*>> \pi(U,w) \\
@V{(i_2)}_*VV @VV{(j_1)}_*V\\
\pi(V,w) @>{(j_2)}_*>> \pi(X,w)  \\
\end{CD}\]
Here it is possible to interpret $\pi(X,w)$ as the free product with amalgamation of $ \pi(U,w)$ and $ \pi(V,w)$ so that, given group presentations:
$$\pi(U,w)=\langle u_{1},\cdots ,u_{k} \ | \ \alpha _{1},\cdots ,\alpha _{l}\rangle=\langle S_1 \ | \ R_1 \rangle;$$ 
$$\pi(V,w)=\langle v_{1},\cdots ,v_{m} \ | \ \beta _{1},\cdots ,\beta _{n}\rangle=\langle S_2 \ | \ R_2 \rangle;$$
$$\pi(U\cap V,w)=\langle w_{1},\cdots ,w_{p} \ | \ \gamma _{1},\cdots ,\gamma _{p}\rangle=\langle S \ | \ R_3 \rangle ;$$
and one can  describe $\pi(X,w)$ in terms of generators and relators 
$$\pi(X,w)=\langle u_{1}, \ldots ,u_{k},v_{1},\cdots ,v_{m} \  |  $$
$$\alpha_1, \ldots ,\alpha_l, \beta_1, \ldots ,\beta_n, {(i_1)}_* (w_1) ({(i_2)}_*(w_{1}))^{-1}, \ldots , {(i_1)}_*(w_{p})({(i_2)}_*(w_{p}))^{-1} \rangle$$
$$=\langle S_1 \cup S_2 \ | \ R_1 \cup R_2 \cup R_S \rangle,$$ 
where $R_S$ is the set of words of the form $(i_1)_*(a)((i_2)_*(a))^{-1}$ with $a \in S$. See details in \cite[Chapters 23, 24, 25 and 26]{kosnio}. There is an alternative way to view a group presentation and that is as a quotient group of another group; it is well known fact in combinatorial group theory \cite{MaKaSo}. In order to give a lemma that shows us this let us first recall that the \textit{normal subgroup of a group generated by a set of elements} is the smallest normal subgroup containing these elements, or equivalently, it is the subgroup generated by the set of elements and their conjugates. 

\begin{lemma}\label{lem3}
Let $G$ have the presentation $\langle a,b,c,... \ | \ P,Q,R,... \rangle,$ where  $a,b,c, ...$ are generators of $G$ and $P=P(a,b,c, ...),Q=Q(a,b,c, ...),R=R(a,b,c, ...), ...$ are words in $G$ that give us the relators. Let $N$ be the normal subgroup of $G$ generated by the words $S(a,b,c,...), T(a,b,c,...), ...$ in $G$, then the quotient group $G/N$ has the presentation $\langle a,b,c, ... \ | \ P,Q,R, ...,S,T,... \rangle. $
\end{lemma}

\begin{proof}
The proof and details  can be found in \cite[Theorem 2.1]{MaKaSo}.
\end{proof}

In particular, we can describe the quotients of the free groups and this gives a powerful tool, in order to think at arbitrary groups via appropriate quotients of free groups. 

\begin{corollary}\label{lem3bis}
If $F$ is the free group on $a,b,c, ...$ and $N$ is the normal subgroup of $F$ generated by $P(a,b,c, ...)$, $Q(a,b,c, ...)$, $R(a,b,c, ...)$, ..., then $F/N=\langle a,b,c, ... \ | \ P,Q,R, ... \rangle . $
\end{corollary}

\begin{proof}
See  \cite[Corollary 2.1]{MaKaSo}.
\end{proof}

We reprove \cite[Theorem 19.1]{DoerHawk} via commutative diagrams, since the logic will be useful later on.

\begin{lemma} \label{lem1} 
Let $G_1,G_2$, and $H$ be groups, let $\varepsilon_i :G_i \to H$ be an epimorphism, and write $K_i=\ker(\varepsilon_i)(i=1,2)$. Let $D=G_1 \times G_2$ and $G=\{(g_1,g_2): g_i \in G_i \mbox{ and } \varepsilon_1(g_1)=\varepsilon_2(g_2) \}.$
Then $G$ is a subgroup of $D$, and there exist epimorphisms $\delta_i : G \to G_i$ such that
\item[(a)] $\ker(\delta_1)=G \cap G_2 \cong K_2$ and $\ker(\delta_2)=G \cap G_1 \cong K_1$,
\item[(b)] $\ker(\delta_1) \ker(\delta_2)=K_1 \times K_2$, and
\item[(c)] $G/(K_1 \times K_2) \cong H$.
\end{lemma}

\begin{proof}From the assumptions, we may consider the following commutative diagram, where $\delta_i = \pi_i \circ \iota$ is clearly an epimorphism, and  check that the thesis is satisfied when $K_1=\ker (\delta_2)$ and $K_2=\ker (\delta_1)$;

\adjustbox{scale=1.5,center}{\begin{tikzcd}
G
\arrow[twoheadrightarrow, rr, bend left, "\delta_2"]
\arrow[twoheadrightarrow, dr,   "\delta_1"]
\arrow[hook, r, "\iota"] 
\arrow[hook, d, "\iota"]
& G_1 \times G_2 \arrow[twoheadrightarrow, r, "\pi_2"] \arrow[twoheadrightarrow, d, "\pi_1"]
& G_2 \arrow[twoheadrightarrow, d, "\varepsilon_2"] \\
G_1 \times G_2 \arrow[twoheadrightarrow, r, "\pi_1"]
& G_1 \arrow[twoheadrightarrow, r, "\varepsilon_1"]
& H \end{tikzcd}}
\end{proof}

Following \cite{DoerHawk}, the subgroup $G$ of $G_1 \times G_2$ constructed in Lemma \ref{lem1} is  a \textit{subdirect product of $G_1$ and $G_2$}. On the other hand, 

\begin{definition}\label{cp}
An arbitrary  group $C$ is a \textit{central product} of its subgroups $A$ and $B$, if $C=AB$ and $[A,B]=1$. 
\end{definition}

We use the notation $C=A \circ B$ when we are in a  situation like Definition \ref{cp}. Note that   both $A$ and $B$ are normal in $C$ in Definition \ref{cp}; moreover  $A \cap B \leqslant Z(A) \cap Z(B)$, where $Z(A)=\{a \in A \ | \ ax=xa \ \ \forall x \in A\}$ denotes the center of $A$.

\begin{lemma} \label{lem1bis} 
Let $A,B$ be subgroups of a group $G$, $D=A \times B$  and  $\bar{A}=A \times 1$ and $\bar{B}=1 \times B$. Then the following statements are equivalent:
\begin{itemize}
\item[(a)] $G$ is a central product of $A$ and $B$;
\item[(b)] There exists an epimorphism $\varepsilon: D \to G$ such that $\varepsilon(\bar{A})=A$, $\varepsilon(\bar{B})=B$; 
\end{itemize}
In particular, if $G$ is a central product of $A$ and $B$, then $G$  is  a subdirect product of $A$ and $B$.
\end{lemma}

\begin{proof} This proof can be found in \cite[Lemma 19.4]{DoerHawk}, but we report it with a different argument which we will be useful later on.(a) $\Rightarrow$ (b): If $G=AB$ with $[A,B]=1$ then the map 
$\varepsilon: (a,b) \mapsto ab$
is clearly an epimorphism with the required properties. (b) $\Rightarrow$ (a): Since $\varepsilon$ is an epimorphism, we have $G=\varepsilon(D)=\varepsilon(A \times B)=\varepsilon(\bar{A}\bar{B})=\varepsilon(\bar{A})\varepsilon(\bar{B})=AB.$ Moreover, $[A,B]=[\varepsilon(\bar{A}),\varepsilon(\bar{B})]=\varepsilon([\bar{A},\bar{B}])=1$. So the first part of the result follows. In particular, we apply Lemma \ref{lem1} with $G_1=A$ and $G_2=B$, so the second part of the result follows.
\end{proof}

Lemma \ref{lem1bis} basically says that a group which can be written as the central product of two groups $A$ and $B$ must be necessarily isomorphic to a quotient of $A \times B$.

\subsection{Applications to the group of Wolfgang Pauli} 
Now one can focus on  the Pauli group $P$ and find an equivalent presentation for \eqref{presentation1}, involving $Q_8$ and $\mathbb{Z}(4)$.

\begin{lemma}\label{equivpresentation} The group $P$ can be presented by 
$$P=\langle u, xy, y \ | \ u^4=x^2=1, u^2=y^2, uy=yu, yx=xy, x^{-1}ux=u^{-1}\rangle.$$
Moreover $P=Q_8 \circ \mathbb{Z}(4)$, where $$Q_8=\langle u,xy \ | \ u^4=1, u^2=(xy)^2, (xy)^{-1}u(xy)=u^{-1} \rangle \ \  and  \ \ \mathbb{Z}(4)=\langle y \ | \ y^4=1 \rangle.$$
\end{lemma}

\begin{proof}Consider $P$ as in \eqref{presentation1} and define $$u=XY, \ \ x=Y \ \ \mbox{and} \ \ y=XYZ.$$
Then we need to show that the relations in \eqref{presentation1} can be generated by the relations in the thesis.

First of all we will derive the equations $X^2=Y^2=Z^2=1$. We begin to note that
\begin{equation}\label{first} 
x^2 =1 \ \Longrightarrow \ Y^2=1,
\end{equation}
so one equation in\eqref{presentation1} is obtained and \eqref{first} allows us to conclude that
\begin{equation}\label{second} 
x^{-1}ux =u^{-1} \ \Longrightarrow \ Y^{-1}XYY=(XY)^{-1} \ \Longrightarrow \ Y^{-1}XYY=Y^{-1}X^{-1}
\ \Longrightarrow \ XYY=X^{-1}
\end{equation}
$$ XXYY=1  \ \Longrightarrow \ X^2Y^2=1 \ \Longrightarrow \ X^21=1
\ \Longrightarrow \ X^2=1$$
so a second equation in \eqref{presentation1} is obtained.
Now we need to show that
 \begin{equation}\label{third} 
 uy=yu \ \Longrightarrow \  XYXYZ=XYZXY \ \Longrightarrow  \ XYZ=ZXY,
\end{equation}
in order to derive the following equation
 \begin{equation}\label{four} 
u^2 =y^2 \ \Longrightarrow \ XYXY=XYZXYZ \ \Longrightarrow \ XY=ZXYZ 
\end{equation}
$$ \Longrightarrow \ XY=XYZZ \ \Longrightarrow \
1=Z^2$$
So we have shown until now that $X^2=Y^2=Z^2=1$.

Now we go ahead to show the equations $(YZ)^4=(ZX)^4=(XY)^4=1$. We begin with
 \begin{equation}\label{five} 
u^4=1 \ \Longrightarrow \ (XY)^4=1,
\end{equation}
then with help of \eqref{third} we derive
 \begin{equation}\label{five} 
yx=xy \ \Longrightarrow \  XYZY=YXYZ \ \Longrightarrow \ ZXYY=YXYZ\end{equation}
$$ \ \Longrightarrow \ ZXYY=YZXY \ \Longrightarrow \ ZXY=YZX.$$
Now \eqref{third} allows us to have two more equations, namely
 \begin{equation}\label{six} 
  ZXY=YZX \ \Longrightarrow \   ZXYX=YZXX  \ \Longrightarrow \ ZXYX=YZ,
\end{equation}
 \begin{equation}\label{seven} 
  ZXY=YZX  \ \Longleftrightarrow \   ZX=YZXY \ \Longrightarrow \  ZX=YXYZ.
\end{equation}
Therefore we find
 \begin{equation}\label{eight} 
(YZ)^4=YZYZYZYZ=ZXYXZXYXZXYXZXYX=(ZXY)XZXYXZXYXZXYX
\end{equation}
$$=(YZX)XZXYXZXYXZXYX=YZXXZXYXZXYXZXYX=YXYXYZZXYX$$
$$=YZZXYXZXYXZXYX=Y(ZZ)XYXZXYXZXYX=YXYXZXYXZXYX$$
$$=YXYX(ZXY)XZXYX=YXYX(YZX)XZXYX=YXYXYZXXZXYX$$
$$=YXYXYZ(XX)ZXYX=YZ(XX)ZXYXZXYXZXYX$$
$$=YXYXY(ZZ)XYX=YXYXYXYX=Y(XYXYXY)X$$
$$=Y(XY)^3X=Y(XY)^{-1}X=YY^{-1}X^{-1}X=1$$
and in analogy  $(ZX)^4=1$. This allows us to conclude that \eqref{presentation1} is equivalent to the presentation in the thesis. 
Now we consider the map $$\varepsilon : (a,b )\in Q_8 \times \mathbb{Z}(4) \mapsto ab \in P$$
and  we note that $xy=yx$ in $P$, hence for all $\alpha, \beta, \gamma, \delta \in \{0,1,2,3\}$ we have
$$\varepsilon(u^\alpha, y^\beta)  \varepsilon({(xy)}^\gamma, y^\delta) = (u^\alpha y^\beta) ({(xy)}^\gamma y^\delta)$$
$$=u^\alpha y^\beta x^\gamma y^{\gamma+\delta}=u^\alpha  x^\gamma y^{\beta+\gamma+\delta}=  \varepsilon(u^\alpha{(xy)}^\gamma, y^{\beta+\delta})=\varepsilon((u^\alpha, y^\beta)({(xy)}^\gamma, y^\delta)),
$$
which is enough to conclude that $\varepsilon$ is homomorphism of groups, because we checked on the generic generators of $Q_8 \times \mathbb{Z}(4)$. Finally, $\varepsilon$ is surjective by construction, so $P=Q_8 \circ \mathbb{Z}(4)$ by Lemma \ref{lem1bis}.
\end{proof}

Viceversa the notion of central product can allow us to construct presentations. 

\begin{lemma}\label{lem2}If $G=A \circ B$ with $A=\langle S \ | \ R_A \rangle $ and $B=\langle T \ | \ R_B \rangle $ are presentations for $A$ and $B$, then $G=\langle S \cup T \ | \ R_A \cup R_B \cup R_C \cup R_\varepsilon \rangle,$ where $R_C=\{ (a,b)\in S \times T \ | \ ab=ba \}$ and $R_\varepsilon=\{  (a,b)\in S \times T \ | \ ab=1 \}$.
\end{lemma}

\begin{proof}
From Lemma \ref{lem1bis} if we have the epimorphism $\varepsilon: (a,b) \in A \times B \mapsto ab \in G$, then $G=(A \times B)/\ker(\varepsilon)$. A  presentation for $A \times B$ is of the form $\langle S \cup T \ | \ R_A \cup R_B \cup R_C \rangle$, because in the definition of $A \times B$ we require $[A,B]=1$. Now $\ker (\varepsilon)$ induces the additional relation $R_\varepsilon$ and so $(A\times B)/\ker(\varepsilon)$ is presented as claimed, because of Lemma \ref{lem3}.
\end{proof}

\section{Proof of  Theorem 1.1}

We begin to  prove Theorem \ref{maintheorem}.

\begin{proof}Case (i). We refer to \cite[Chapter 23]{kosnio} for the formulation of the theorem of Seifert and Van Kampen, and the corresponding  terminology has been reported in Paragraph 3.1 above exactly as in \cite[Chapter 23]{kosnio}.  First of all we note that we are in the assumptions of the theorem of Seifert and Van Kampen because of Lemmas \ref{superlegalaction} and \ref{superlegalactiontwo} and Corollary \ref{orbitspaces}. In addition, these results show that  $U$, $V$, $U \cap V$, $X$ are compact and path connected spaces. Therefore $\pi(X,x_0)$ is not dependent on the choice of $x_0 \in U \cap V$. Now we construct $\pi(X)$ directly:  
$$\pi(U)=Q_8=\langle u,xy \ | \ u^4=1, u^2=(xy)^2, (xy)^{-1}u(xy)=u^{-1} \rangle=\langle S_1 \ | \ R_1 \rangle,$$ 
$$\pi(V)=\mathbb{Z}(4)=\langle y \ | \ y^4=1 \rangle=\langle S_2 \ | \ R_2 \rangle,$$ 
$$\pi(U \cap V)=\mathbb{Z}(2)=\langle u^2 \ | \ u^4=1 \rangle \subseteq \pi(U),$$
$$ \pi(U \cap V)=\mathbb{Z}(2)=\langle y^2 \ | \ y^4=1 \rangle \subseteq \pi(V)$$ and one can check that 
$$(i_1)_* :  u^2 \in \pi(U \cap V) \subseteq \pi(U) \mapsto (i_1)_*(u^2)=u^2 \in \pi(U),$$
$$(i_2)_* : y^2 \in \pi(U \cap V) \subseteq \pi(V) \mapsto (i_2)_*(y^2)=y^2 \in \pi(V)$$ are homomorphisms in accordance to the theorem of Seifert-Van Kampen. Moreover $$S=\{u^2 \ | \ u \in \pi(U)\}=\{y^2 \ | \ y \in \pi(V)\}$$ is the set of generators of $\pi(U \cap V)$, but $\mathbb{Z}(2)$  has only one generator so we may deduce that $$R_S = \{u^2y^2=1 \ | \ u\in \pi(U), y\in \pi(V)\}=\{u^2=y^2 \ |  \ u\in \pi(U), y\in \pi(V)\}$$ are the relations on $\pi(X)$ induced by $\pi(U \cap V)$ from the theorem of Seifert and van Kampen.

We conclude that $$\pi(X)=\langle S_1 \cup S_2 \ | \ R_1 \cup R_2 \cup R_S \rangle = \langle u,xy, y \ | \ u^4=y^4=1, u^2=(xy)^2, (xy)^{-1}u(xy)=u^{-1},  u^2=y^2  \rangle.$$ Let $N$ be the normal subgroup generated by  $[S_1,S_2]=\{s_1^{-1}s_2^{-1}s_1s_2 \ | \ s_1 \in S_1, s_2 \in S_2\}$ in $\pi(X)$. By Lemma \ref{lem3} we get the following presentation for the quotient group:   
 
\begin{equation} \label{seifquo}
\pi(X)/N=\langle S_1 \cup S_2 \ | \ R_1 \cup R_2 \cup R_3 \cup R_S \rangle
\end{equation} 
 $$= \langle u,xy, y \ | \ u^4=y^4=1, u^2=(xy)^2, (xy)^{-1}u(xy)=u^{-1},  u^2=y^2, uy=yu,  xyy=yxy  \rangle$$
 where 
 $$R_3=\{s_1s_2=s_2s_1 \ | \ s_1 \in S_1, s_2 \in S_2\},$$ and we claim that \eqref{seifquo} is equivalent to 
 the following presentation of Lemma \ref{equivpresentation}
\begin{equation} \label{pauli}
P=\langle u, xy, y \ | \ u^4=x^2=1, u^2=y^2, uy=yu, yx=xy, x^{-1}ux=u^{-1}\rangle.
\end{equation}

Since \eqref{seifquo} and \eqref{pauli} have the same  generators, the relations in \eqref{pauli} will be deduced from the relations in \eqref{seifquo} and viceversa. Firstly we consider \eqref{seifquo} and note that
$$xyy=yxy \iff xy^2=yxy \iff xy^2y^3=yxyy^3 \iff xy^5=yxy^4 $$
$$\iff xyy^4=yxy^4 \iff xy=yx$$
and so $x$ and $y$ commute. Similarly one can see that $u$ commutes with $y$.
Secondly we have
$$u^2=(xy)^2 \iff u^2=x^2y^2 \iff u^2=x^2u^2 \iff u^2u^2=x^2u^2u^2 \iff u^4=x^2u^4 \iff 1=x^2.$$
Thirdly we note that
$$(xy)^{-1}u(xy)=u^{-1} \iff x^{-1}y^{-1}uxy=u^{-1} \iff x^{-1}y^{-1}uyx=u^{-1} $$$$\iff x^{-1}y^{-1}yux=u^{-1} \iff x^{-1}ux=u^{-1}$$
Finally all the other relations in \eqref{seifquo} are clearly present in Lemma \ref{equivpresentation} so our claim follows and
we may conclude that   $P \cong \pi(X)/N$. Now we apply Lemma \ref{lem2} and realize that $P=\pi(U) \circ \pi(V)$.

\bigskip

Case (ii). From \cite[Exercise 11.2 (d)]{kosnio} both $U$ and $V$ are 3-manifolds and so \cite[Exercise 11.5 (b)]{kosnio} implies that $U \# V$ is a 3-manifold. 
Now we observe from \cite{bubuiv} that the action of $Q_8$ on the Riemannian sphere $S^3$ (with the round metric $d_{S^3}$ that is derived from the Riemannian metric $$ds^2=\frac{4 \|dx\|^2}{(1+\|x\|)^2},$$ where $\|dx\|^2$ is the usual Riemannian metric on $\mathbb{R}^2$) produces the Riemannian space of orbits $U=S^3/Q_8$, with canonical quotient map $p_U: S^3 \to U$ with the induced distance function on $u_1,u_2 \in U$,
$$d_U(u_1,u_2) = \inf\limits_{x \in p_U^{-1}(u_1),y \in p_U^{-1}(u_2)} d_{S^3}(x,y)$$ 
and similarly this happens for $V$ with $d_V$, because the actions are free and properly discontinuous (see Lemma \ref{superlegalaction}), as well as $Q_8$ and $\mathbb{Z}(4)$ being groups of isometries on $S^3$ (see \cite[Remark 41]{thesis}). 

According to the terminology of \cite[Definition 22]{thesis},  a pair $(M,\Gamma)$ where $M$ is a Riemannian manifold and $\Gamma$ is a (proper) discontinuous group of isometries acting effectively on $M$ is called ``good Riemannian orbifold''. The underlying space of the orbifold is $M/\Gamma$. In the case of a good Riemannian orbifold $(M,\Gamma)$ it follows that for $x,y \in M/\Gamma$, $$d(x,y)=d_M(\pi^{-1}(x),\pi^{-1}(y)) := \inf\limits_{\tilde{x} \in \pi^{-1}(x),\tilde{y} \in \pi^{-1}(y)} d_M(\tilde{x},\tilde{y})$$
and this is exactly the situation we have here with $M=S^3$ and $\Gamma$ as indicated in Lemmas \ref{superlegalaction} and \ref{superlegalactiontwo}. This means that we can define a metric on the disjoint union of the subspaces $U'$ and $V'$ (which are $U$ and $V$ without the open balls of a connected sum construction) by:
\[   d'(x,y) = \left\{
\begin{array}{ll}
      d_{U'}(x,y), & if x,y \in U' \\
      d_{V'}(x,y), & if x,y \in V' \\
      \infty, & \text{otherwise}.
\end{array} 
\right. \]
Using this metric we can endow $U \# V$ with the structure of Riemannian space and consider the quotient semi-metric on $ U \# V$ in the sense of   \cite[Definition 3.1.12]{bubuiv} 
$$d_R(x,y)=\inf \{\sum\limits^k_{i=1} d'(p_i,q_i) \ | \ p_1=x, q_k=y, k \in \mathbb{N}  \},$$ 
where $R$ is the equivalence relation induced by $\#$. In particular,  $d_R(x,y)$ is a Riemannian metric on $U \# V$, since $U' \cup V'$ is compact (see \cite[Exercise 3.1.14]{bubuiv}).

From \cite[Exercise 26.6(c)]{kosnio} we have  $\pi(U \# V) \simeq \pi(U) \ast \pi(V)$ and so $$\pi(U \# V)=Q_8 \ast \mathbb{Z}(4)=\langle u,xy,y \ | \ u^4=1, u^2=(xy)^2, (xy)^{-1}u(xy)=u^{-1}, y^4=1\rangle.$$
Imposing the relations  $uy=yu,xyy=yxy,u^2=y^2$, we consider the existence of a normal subgroup $L$ in $\pi(U \# V)$ by Lemma \ref{lem3}, and  get again the presentation \eqref{seifquo}, which we have seen to be equivalent to that in Lemma \ref{equivpresentation}. Therefore  $P=\pi(U) \circ \pi(V)$ and Case (ii)  follows completely.

\end{proof}

It is appropriate to make some comments here on the choice of $S^3$ in the context of the present investigations. Looking at \cite[Chapters 1, 4]{hatcher} or at \cite[Chapter 29]{kosnio}, we know that a path connected space whose fundamental group is isomorphic to a given group $G$ and which has a contractible universal covering is a \textit{$K(G,1)$ space},  also known as \textit{Eilenberg-MacLane space} (of type one). Roughly speaking, these topological spaces answer the problem of realizing a prescribed group $G$ as fundamental group of an appropriate path connected space $X$. Eilenberg-MacLane spaces are well known in algebraic topology, so one could wonder why we didn't use them.

A first motivation is that $X$ is constructed as a polyhedron (technically, $X$ is a \textit{cellular complex}, see \cite[Chapter 0]{hatcher}) so it is in general difficult to argue whether $X$ possesses a Riemannian structure or not, when we construct $X$ with the method of Eilenberg and Mac Lane. Our direct observation of the properties of $S^3$ allows us to find an interesting behaviour, looking at the final part of the proof of Theorem \ref{maintheorem}: the Pauli group can be constructed in the way we made  and it can be endowed by a Riemannian structure, arising from the Riemannian metric which we introduced.

Then we come to a second motivation which justifies our approach via $S^3$. Finding a Riemannian structure, in connection with a group, has a relevant meaning in several models of quantum mechanics. For instance, Chepilko  and  Romanenko \cite{cr} produced a series of contributions, illustrating how certain processes of quantization and some sophisticated variational principles may be easily understood in presence of 
Riemannian manifolds and groups. In this perspective one can also look at \cite{pv}, which shows again the strong simplification of the structure of the hamiltonian in presence of models where we have both a Riemannian manifold and a group of symmetries. There are of course more examples on the same line of research and this shows  the relevance of our construction once more.

Last (but not least) we know from \cite[Corollary 9.59, (iv), The Sphere Group Theorem]{hofmor} that:  the only groups on a sphere are either the two element group $S^0=\{1,-1\}$, or the circle group $S^1$, or the group $S^3$ of quaternions of norm one (note that $S^3$ may be identified with $ SU(2)$ and we gave details in Section 2 of the present paper). Note that all of them are compact connected Lie groups (apart from $S^0$ which is disconnected) and we know from Lie theory that a  structure of differential manifold can be introduced on compact connected Lie groups via the Baker--Campbell--Hausdorff formalism \cite[Chapters 5 and 6]{hofmor}. 
Therefore the relevance of $S^3$ appears again among all the possible spheres $S^n$, which possesses a group structure, because it allows us to produce also a differential structure on $P$, realized in the way we made.

We end with an observation which was surprising in our investigation. 

\begin{remark}\label{rem4} In \cite{rocchetto}, it was discussed the decomposition $P=D_8 \circ \mathbb{Z}(4)$ with an appropriate study of the abelian subgroups of $P$ arising from this decomposition; so not only $P=Q_8 \circ \mathbb{Z}(4)$ is true. On the other hand,  $D_8$ cannot act freely on any $S^n$ by Remarks \ref{rem1} and \ref{rem2}, so  one cannot replace the role of $Q_8$ with $D_8$ in the proof of Theorem \ref{maintheorem}, even if algebraically $Q_8$ and $D_8$ are very similar. This is a further element of interest for the methods that we offered here.
\end{remark}

\section{Proof of  Theorem 1.2 and connections with physics}

In \cite{bagrus1,bagrus2} connections between  some purely algebraic results and physics, and Quantum Mechanics in particular, have been considered. The bridge between the two realms was provided by the so-called {\em pseudo-bosons}, studied intensively in { a series of recent contributions \cite{bagweak, baginbagbook} }. In particular, we refer to \cite{ baginbagbook} for a (relatively) recent review. In this perspective, it is natural to see if and how the Pauli group $P$ is related to mathematical objects which, in some sense, are close to pseudo-bosons. As we will show here, this is exactly the case: pseudo-fermions, which are a sort of two-dimensional version of pseudo-bosons, can be used to describe the elements $X$, $Y$ and $Z$ of $P$ in \eqref{paulimatricesxyz}, and, because of this, they appear to have a direct physical meaning. We refer to \cite{baginbagbook,bagpf1,bagpf2} for the general theory of pseudo-fermions, and to \cite{baggarg} for several physical applications of these excitations. We should also mention that similar {\em pseudo-particles} have been considered by several authors in the past decades. We only cite here few contributions, \cite{male,moh,gre}.  To keep the paper self-contained, we have also given a crash course on pseudo-fermions in  Appendix. 

The idea is quite simple: we consider two operators $a$ and $b$ on the Hilbert space $\Hil=\mathbb{C}^2$ satisfying the following rules:
\begin{equation} \label{51}
\{a,b\}=ab+ba=\1, \qquad a^2=b^2=0,
\end{equation} 
where $\{a,b\}$ is the anticommutator between $a$ and $b$, and $\1$ is the identity operator on $\Hil$. Of course, if $b=a^\dagger$, the adjoint of $a$, (\ref{51}) returns the so-called {\em canonical anti-commutation relations}, CAR. The operators $a$ and $b$ are the basic ingredients now to define the following operators on $\Hil$:
\begin{equation} \label{52}
\mu_1=b+a, \qquad \mu_2=i(b-a), \qquad \mu_3=[a,b]=ab-ba.
\end{equation} 
In particular, the square brackets are called {\em the commutator} between $a$ and $b$. Incidentally we observe that $\{a,b\}=\1$, because of the (\ref{51}).  The main result of this section is that the set $P_\mu=\{\mu_1,\mu_2,\mu_3\}$ is a concrete realization of the Pauli group. The proof of this claim is based on several identities which can easily be deduced out of (\ref{51}). More in details, we have that
\begin{equation} \label{53}
\mu_j^2=\1, \quad j=1,2,3, \qquad \mbox{ and } \qquad \mu_1\mu_2=i\mu_3, \quad  \mu_2\mu_3=i\mu_1, \quad  \mu_3\mu_1=i\mu_2.
\end{equation} 
In fact we have
$$
\mu_1^2=(b+a)^2=b^2+ba+ab+a^2=\{a,b\}=\1,
$$
since $a^2=b^2=0$ and $\{a,b\}=\1$. Similarly we can check that $\mu_2^2=\1$. Slightly longer is the proof that  $\mu_3^2=\1$. We have
$$
 \mu_2^2=(ab-ba)^2=abab+baba-abba-baab=(\1-ba)ab+(\1-ab)ba=ab+ba=\1,
$$
where we have used several times the equalities in (\ref{51}). The proof of the other equalities in (\ref{51}) is similar, and will not be given here. However, these equalities are relevant to prove that, indeed, $P_\mu$ is a Pauli group. Infact, we have $(\mu_1\mu_2)^4=(i\mu_3)^4=i^4(\mu_3^2)^2=\1^2=\1$. Similarly we can check that $(\mu_2\mu_3)^4=(\mu_3\mu_1)^4=\1$, so that our claim is proved.

$P_\mu$ is not the only Pauli group which can be constructed out of pseudo-fermionic operators. In fact, $P_\rho=\{\rho_j=\mu_j^\dagger, \, j=1,2,3\}$ is also a Pauli group, meaning with this that the following equalities are all satisfied: $$ \rho_j^2=\1,  \ \ j=1,2,3,   \ \ \mbox{and} \ \ \rho_1\rho_2=i\rho_3, \ \   \rho_2\rho_3=i\rho_1, \ \ \mbox{and} \ \ \rho_3\rho_1=i\rho_2.$$

In \cite{baggarg,baginbagbook} many applications of pseudo-fermions to physics have been discussed. This suggests that the theory developed here is somehow connected to physics, and in particular to Quantum Mechanics, as the following result connected to a two-levels atom with damping clearly shows.

In 2007 \cite{tripf}, an effective non self-adjoint hamiltonian describing a two level atom interacting with an electromagnetic field was analyzed in connection with pseudo-hermitian systems, \cite{ben}. Later (see \cite{bagpf1}) it has been shown that this model can be  rewritten in terms of pseudo-fermionic operators, and, because of what discussed in this section, in terms of Pauli groups.

\begin{proof}[Proof of Theorem 1.2]The starting point is the Schr\"odinger
equation 

\begin{equation} i\dot\Phi(t)=H_{eff}\Phi(t), \qquad H_{eff}=\frac{1}{2}\left(
\begin{array}{cc}
	-i\delta & \overline{\omega} \\
	\omega & i\delta \\
\end{array}
\right).
\label{triex}\end{equation} Here $\delta$ is a real quantity, related to the decay rates for the two levels, while the complex parameter $\omega$
characterizes the radiation-atom interaction. We refer to \cite{tripf} for further details. It is clear that $H_{eff}\neq H_{eff}^\dagger$. It
is convenient to write $\omega=|\omega|e^{i\theta}$. Then, we introduce the  operators
$$
a=\frac{1}{2\Omega}\left(
\begin{array}{cc}
-|\omega| & -e^{-i\theta}(\Omega+i\delta) \\
e^{i\theta}(\Omega-i\delta) & |\omega| \\
\end{array}
\right), \quad
b=\frac{1}{2\Omega}\left(
\begin{array}{cc}
-|\omega| & e^{-i\theta}(\Omega-i\delta) \\
-e^{i\theta}(\Omega+i\delta) & |\omega| \\
\end{array}
\right).
$$
Here $$\Omega=\sqrt{|\omega|^2-\delta^2},$$ which we will assume here to be real and strictly positive. A direct computation shows that
$\{a,b\}=\1$, $a^2=b^2=0$. Hence $a$ and $b$ are pseudo-fermionic operators. Moreover, $H_{eff}$ can be written in terms of these operators as
$$H_{eff}=\Omega\left(ba-\frac{1}{2}\1\right).$$ It is now easy to identify the elements of $P_\mu$, using (\ref{52}). We get
$$
\mu_1=\frac{1}{\Omega}\left(
\begin{array}{cc}
-|\omega| & -i\delta e^{-i\theta} \\
-i\delta e^{i\theta} & |\omega| \\
\end{array}
\right), \quad \mu_2=i\left(
\begin{array}{cc}
0 &  e^{-i\theta} \\
 e^{i\theta} & 0 \\
\end{array}
\right), \quad \mu_3=\frac{1}{\Omega}\left(
\begin{array}{cc}
i\delta & -|\omega| e^{-i\theta} \\
-|\omega| e^{i\theta} & -i\delta \\
\end{array}
\right), 
$$
which are therefore a (non-trivial, and physically motivated) representation of the Pauli group. In terms of these operators $H_{eff}$ acquires the following particularly simple expression: $$H_{eff}=-\frac{\Omega}{2}\,\mu_3.$$

The elements $u$, $xy$ and $y$ in Lemma \ref{equivpresentation} can be computed and turns out to be
$$
u=i\mu_3=\frac{i}{\Omega}\left(
\begin{array}{cc}
i\delta & -|\omega| e^{-i\theta} \\
-|\omega| e^{i\theta} & -i\delta \\
\end{array}
\right),\qquad xy=i\mu_2=-\left(
\begin{array}{cc}
0 &  e^{-i\theta} \\
e^{i\theta} & 0 \\
\end{array}
\right), \qquad y=i\1.
$$
This produces an interesting consequence for this model. Since $y$ is proportional to the identity element, and since $Q_8$ only contains $\mu_2$ and $\mu_3$, we  interpret the elements of $\mathbb{Z}(4)$ as the constants of motion of the physical system described by $H_{eff}$, or by its generalized form $$H_{eff}'=-\frac{\Omega}{2}\,\mu_3+\alpha\mu_2,$$ for all possible real $\alpha$. This is an interesting, and somehow unexpected, feature of the model: going from a larger $P_\mu$ to a smaller group $Q_8$ does not affect at all the dynamical aspects (i.e., the generator of the time evolution) of the system, since these  are all contained in $Q_8$.
\end{proof}

\begin{remark} It is useful to stress that, while the algebraic construction discussed here is totally independent of what deduced in the first part of the paper, the last part of the proof above, and in particular the role of $\mathbb{Z}(4)$ and $Q_8$ for this specific system, appears quite interesting, and clearly open the possibility that similar results can also be found in other quantum mechanical models described in terms of the Pauli group. The natural question, which we will consider in a future paper, is whether for this kind of systems the factor group $\mathbb{Z}(4)$ (appearing in $P=Q_8 \circ \mathbb{Z}(4)$) always contains the physical constants of motion. In fact, we do not expect this is a completely general feature, but we believe it can be true under some additional, and reasonable, assumptions. We will return on this aspect in the conclusions.
\end{remark}

We end this section by noticing that the $\mu_j$ return $X$, $Y$ and $Z$ in \eqref{paulimatricesxyz} under suitable limiting conditions on the parameters: if $\theta,\delta\rightarrow0$, then $u\rightarrow-iX$ and $xy\rightarrow-iY$. This shows clearly that our present representation extends that of the previous sections.

\section{Conclusions}
 Some properties, which have been noted  in Theorem \ref{maintheorembis}, may be extended to mathematical models, where pseudo-fermionic operators, or generalizations of them, are involved. This is related with the structure of $P$ in Theorem \ref{maintheorem} and Lemma \ref{equivpresentation}.

Look at the proof of Theorem \ref{maintheorem} and at the structure of $P=Q_8 \circ \mathbb{Z}(4) \simeq \pi(U \cup V)/N$, where $U=S^3/Q_8$ and $V=S^3/\mathbb{Z}(4)$ follow the notations of Theorem \ref{maintheorem}. We noted in the final part of the proof Theorem \ref{maintheorembis}  that going from a larger $P_\mu$ to a smaller group $Q_8$ does not affect at all the dynamical aspects (i.e., the generator of the time evolution) of the system, since these  are all contained in $Q_8$. Looking at Remarks \ref{rem1}, \ref{rem2} and \ref{rem3}, we believe that:

\begin{conj}\label{maybe} Groups of the form $A=Q_8 \circ B $, where $B$ is an abelian group containing at most one element of order $2$ may have a construction of the Hamiltonian as we made with $H_{eff}$ in Theorem \ref{maintheorembis}, producing
  the fact that  going from a larger group $A$ to a smaller group $Q_8$ all the dynamical aspects are not affected. 
  \end{conj}

  The behaviour, which we have conjectured above (and shown rigorously in Theorem \ref{maintheorembis}) can be justified, on the basis of the results of Sections 2, 3 and 4.

 If we look at the proof of Theorem \ref{maintheorembis} from a different perspective, we may note that the constants of the motion are somehow unaffected by the operator of central product, when we have $P=Q_8 \circ \mathbb{Z}(4)$. Conjecture \ref{maybe} motivates us to think that the same behaviour happens when we are in presence of an appropriate Hamiltonian and of a group with the structure $A=Q_8 \circ B $, where $B$ can be, for instance, $\mathbb{Z}(2m)$ for any  $m \ge 2$. In fact all such groups act freely on $S^3$.

  Due to \cite[Corollary 9.59, (iv), The Sphere Group Theorem]{hofmor}, which explains the peculiarity of $S^3$ among all $S^n$, we do not think to be reasonable to expect appropriate interpretations of the constants of the motions for Hamiltonians which can be constructed in the same way we did in Theorem \ref{maintheorembis} but with groups of the form $C=D \circ E$ for arbitrary choices of finite groups $D$ and $E$. Maybe one could think at more general frameworks, not involving pseudo-fermionic operators, but then one could loose the information at the level of the physics, while working in the direction of Conjecture \ref{maybe} above, one could get to a significant idea in the mathematical models of quantum mechanics with pseudo-fermions. We hope to give more results in this direction in a near future.
  
\section*{Acknowledgements}

F.B. acknowledges  support  from Palermo University and from the Gruppo Nazionale di Fisica Matematica of the I.N.d.A.M. The other two authors (Y.B. and F.G.R.) thank Shuttleworth Postgraduate Scholarship Programme 2019 and NRF for grants no.  118517 and 113144.

\renewcommand{\theequation}{A.\arabic{equation}}

\section*{Appendix}\label{sectpfs}

The present appendix is meant to make the paper self-contained, by giving some essential definitions and results on pseudo-fermions. We consider two operators $a$ and $b$, acting on the Hilbert space $\Hil=\Bbb C^2$, which satisfy the following rules, \cite{bagpf1}: \be \{a,b\}=\1, \quad
\{a,a\}=0,\quad \{b,b\}=0, \label{FB220}\en where  $\{x,y\}=xy+yx$ is the anti-commutator between $x$ and $y$. We first observe that  a non zero vector $\varphi_0$ exists in $\Hil$ such that $a\,\varphi_0=0$. Similarly,
	a non zero vector $\Psi_0$ exists in $\Hil$ such that $b^\dagger\,\Psi_0=0$.
	This is because the kernels of $a$ and $b^\dagger$ are non-trivial.
\vspace{3mm}

It is now possible to deduce the following results. We first introduce the following non zero vectors \be
\varphi_1:=b\,\varphi_0,\quad \Psi_1=a^\dagger \Psi_0, \label{FB221}\en as well as the non self-adjoint operators \be N=ba,\quad
N^\dagger=a^\dagger b^\dagger. \label{FB222}\en Of course, it makes no sense to consider $b^n\,\varphi_0$ or ${a^\dagger}^n \Psi_0$ for $n\geq2$, since all these vectors are automatically zero. This is analogous to what happens for ordinary fermions. Let now introduce the self-adjoint operators $S_\varphi$ and $S_\Psi$ via their action on a
generic $f\in\Hil$: \be S_\varphi f=\sum_{n=0}^1\left<\varphi_n,f\right>\,\varphi_n, \quad S_\Psi f=\sum_{n=0}^1\left<\Psi_n,f\right>\,\Psi_n.
\label{FB223}\en The following results can be easily proved:

\begin{itemize}
	
	\item[] \be a\varphi_1=\varphi_0,\quad b^\dagger\Psi_1=\Psi_0. \label{FB224}\en
	\item[] \be N\varphi_n=n\varphi_n,\quad N^\dagger \Psi_n=n\Psi_n,  \ \mbox{for} \  n=0,1.\label{FB225}\en
	\item[] If the normalizations of $\varphi_0$ and $\Psi_0$ are chosen in such a way that $\left<\varphi_0,\Psi_0\right>=1$,
	then \be \left<\varphi_k,\Psi_n\right>=\delta_{k,n},    \ \mbox{for} \ k,n=0,1.\label{FB226}\en 
	\item[] $S_\varphi$ and $S_\Psi$ are bounded, strictly positive, self-adjoint, and invertible. They satisfy
	\be \|S_\varphi\|\leq\|\varphi_0\|^2+\|\varphi_1\|^2, \quad \|S_\Psi\|\leq\|\Psi_0\|^2+\|\Psi_1\|^2,\label{FB227}\en 
	\be S_\varphi
	\Psi_n=\varphi_n,\qquad S_\Psi \varphi_n=\Psi_n,\label{FB228}\en for $n=0,1$, as well as $S_\varphi=S_\Psi^{-1}$ and the following 	intertwining relations \be S_\Psi N=N^\dagger  S_\Psi,\qquad S_\varphi N^\dagger =N S_\varphi.\label{FB229}\en

\end{itemize}

 Notice that, being biorthogonal, the vectors of both $\F_\varphi$ and $\F_\Psi$
are linearly independent. Hence $\varphi_0$ and $\varphi_1$ are two linearly independent vectors in a two-dimensional Hilbert space, so that $\F_\varphi$
is a basis for $\Hil$. The same argument obviously can be used for $\F_\Psi$. More than this: both these sets are also Riesz bases. We refer to \cite{baginbagbook} for more details.

\end{document}